\documentclass[preprint,12pt]{elsarticle}

\usepackage{amsmath,amssymb,amsthm,mathtools}
\usepackage{bm}
\usepackage{booktabs}
\usepackage{graphicx}
\usepackage{float}
\usepackage{geometry}
\usepackage[
  colorlinks=true,
  linkcolor=blue,
  citecolor=blue,
  urlcolor=blue
]{hyperref}
\usepackage[capitalize]{cleveref}

\usepackage{lineno}
\journal{Nuclear Physics B}

\newtheorem{theorem}{Theorem}
\newtheorem{proposition}{Proposition}

\theoremstyle{remark}

\newcommand{\Mfp}{\mathcal{M}}
\newcommand{\KA}{K_A}
\newcommand{\Hzero}{H_0}
\newcommand{\VA}{V_A}
\newcommand{\ad}{\operatorname{ad}}
\newcommand{\spec}{\operatorname{spec}}

\newcommand{\Ker}{\operatorname{Ker}}
\newcommand{\Sdweak}{\mathcal{S}_{d,\infty}}

\newcommand{\acrit}{a_{\mathrm{crit}}}
\newcommand{\Scrit}{S_{\mathrm{crit}}}
\newcommand{\Qmin}{Q_{\mathrm{min}}}
\newcommand{\atil}{\tilde a_{\mathrm{crit}}}
\newcommand{\dd}{\,d}
\newcommand{\anp}{a_{\mathrm{np}}}
\newcommand{\Lam}{\Lambda}
\newcommand{\Riesz}{R}
\newcommand{\ket}[1]{\lvert #1\rangle}
\newcommand{\bra}[1]{\langle #1\rvert}

\begin{document}

\begin{frontmatter}

\title{Birman-Schwinger Formulation of the Faddeev-Popov Zero-Mode Problem}

\author{Daniel G. Tedesco} 

\affiliation{organization={ESEHL/PPGENT - UNINTER},
            addressline={Av. Sete de Setembro, 4228}, 
            city={Curitiba},
            postcode={80250-085}, 
            state={Paran\'a},
            country={Brazil}}

\begin{abstract}
After removing the constant adjoint modes associated with global gauge rotations, we recast the Landau-gauge Faddeev-Popov zero-mode problem in Birman-Schwinger form on a periodic domain. The construction yields a self-adjoint Faddeev-Popov realization under the stated regularity assumptions and a normalized operator with a fixed spectral criterion for the first Gribov horizon. The same formulation relates the horizon condition to the ghost resolvent at fixed gauge background while distinguishing fixed-background spectral information from ensemble-averaged propagators. For a periodic transverse background in $SU(2)$ Yang-Mills theory, the zero-mode equation reduces to a Mathieu problem, allowing the horizon threshold to be determined independently through spectral and finite-channel methods. The construction also yields explicit volume dependence for the critical background and its classical action.
\end{abstract}


\begin{highlights} 
\item Faddeev-Popov zero modes are recast in Birman-Schwinger form. 
\item The first Gribov horizon becomes the spectral threshold $-1$. 
\item Cwikel estimates control the singular values of the normalized operator. 
\item The fixed-background ghost dressing is a diagonal resolvent. 
\item A periodic $SU(2)$ background reduces the problem to a Mathieu chain. 
\end{highlights}


\end{frontmatter}



\section{Introduction}
\label{sec:introduction}

Gribov showed that, in a non-Abelian gauge theory, the transversality condition does not uniquely select a representative on each gauge orbit. In the Hamiltonian formulation of Yang-Mills theory in Coulomb gauge, distinct gauge-equivalent configurations may satisfy the same gauge condition, thereby giving rise to what are now known as Gribov copies \cite{Gribov1978,VandersickelZwanziger2012}. Infinitesimally, this loss of uniqueness is associated with zero modes of the corresponding Faddeev-Popov operator, and the same mechanism applies to the covariant Landau gauge considered here. Once the constant adjoint modes generated by global gauge transformations are removed, the first Gribov region $\Omega$ is defined as the set of transverse gauge configurations for which the reduced Faddeev-Popov operator is positive definite. Its boundary, the first Gribov horizon $\partial\Omega$, is reached when the lowest nontrivial eigenvalue of this operator vanishes.

Zwanziger realized the restriction to $\Omega$ in Euclidean Landau gauge through the horizon function, producing a local and renormalizable action that ties the Gribov restriction to the infrared ghost sector \cite{Zwanziger1989,VandersickelZwanziger2012}. The equivalence between the horizon condition and Gribov's no-pole condition was later established to all orders in the standard Landau gauge \cite{CapriDudalGuimaraesPalharesSorella2013}, and subsequent work incorporated condensates, infrared mass scales, BRST control, and gauge-parameter dependence \cite{DudalSobreiroSorellaVerschelde2005,CapriDudalFiorentiniGuimaraesJusto2016,CapriUniversal2018,MintzPalharesPeruzzoSorella2018,DudalFelixPalharesRondeauVercauteren2019,JustoPereiraSobreiro2022}. The geometry behind this construction has also been studied on the quotient of the space of connections by the gauge group, where the orbit-space distance and its infinitesimal metric measure the local separation between physically inequivalent configurations \cite{OrlandMetric1997}, and where gauge-invariant coordinates realize the same quotient with singularities tied to additional harmonic directions and to chart boundaries that contain Gribov-type degeneracies \cite{OrlandGaugeCoordinates2004}. These analyses place the Faddeev-Popov operator inside the local geometry of the projection from connection space to orbit space.

Explicit normalizable zero modes are known in several dimensions and gauges \cite{GuimaraesSorella2011,LandimVilarVenturaLemes2012}, and lattice calculations show that the low-lying spectrum of the Landau-gauge Faddeev-Popov operator controls the infrared ghost propagator and drifts toward the origin as the volume grows \cite{SternbeckIlgenfritzMuellerPreussker2005,Sternbeck2006,Greensite2010,CucchieriMendes2013}. Bloch-wave methods extend this picture beyond the volumes reachable by direct simulation \cite{CucchieriMendesBloch2016}. These results motivate a formulation of the horizon condition in terms of a normalized operator whose spectral threshold is fixed once and for all, rather than in terms of the bare second-order operator whose eigenvalues move with the volume.

For a transverse background $A$ the formal Landau-gauge expression is split as
\begin{equation}
  \Mfp=\Hzero+\VA,\qquad \Hzero=-\partial^2,\qquad \VA=-g\,\ad(A_i)\partial_i,
  \label{eq:intro-split}
\end{equation}
and the constant adjoint scalars are projected out so that $\Hzero$ is strictly positive on the reduced ghost space. For smooth coefficients one may write
\begin{equation}
  \KA=\Hzero^{-1/2}\VA\Hzero^{-1/2},
  \label{eq:intro-K}
\end{equation}
whereas at the critical regularity used below $\KA$ is defined by the corresponding bounded quadratic form. In either formulation the relation with the Faddeev-Popov operator is the form congruence
\begin{equation}
  q_A[u,v]
  =\big\langle \Hzero^{1/2}u,(1+\KA)\Hzero^{1/2}v\big\rangle,
  \label{eq:intro-congruence}
\end{equation}
which yields
\begin{equation}
  0\in\spec(\Mfp)\quad\Longleftrightarrow\quad -1\in\spec(\KA)
  \label{eq:intro-criterion}
\end{equation}
whenever the associated self-adjoint realization is used. Birman-Schwinger reductions trade zero-energy or bound-state questions for spectral conditions on compact operators with quantitative ideal estimates \cite{Birman1961,Schwinger1961,Simon2005,BehrndtTerElstGesztesy2020}. The map $u\mapsto\Hzero^{1/2}u$ is a bijection between the Faddeev-Popov form domain and the reduced $L^2$ space, so the congruence preserves the dimensions of the positive, negative, and null form subspaces even though it does not preserve numerical eigenvalues.

Because $\VA$ is first order, the coefficient scale selected by the form estimate is $A\in L^d$ for $d\ge3$. We prove that the associated $\KA$ belongs to the weak Schatten ideal $\Sdweak$ and satisfies
\begin{equation}
  s_n(\KA)\le C_{d,N}\,g\,\|A\|_{L^d}\,n^{-1/d},
  \label{eq:intro-cwikel}
\end{equation}
using a Cwikel estimate established directly for tori \cite{McDonaldPonge2022}. The dimension $d=2$ is an endpoint of a different type. The constant mode removed from the ghost space must remain projected out in the Birman-Schwinger counting identity, and the resulting compressed Schr\"odinger operator differs from the unprojected comparison by at most one negative direction in the scalar problem. This codimension-one compression does not alter the failure of the uniform  $L^2$ estimate. The positive periodic endpoint is instead controlled by the Orlicz condition $|f|^2\in L\log L$, equivalently $f\in L_\Phi$ with $\Phi(t)=t^2\log(e+t^2)$, through the Cwikel-Solomyak estimate on the torus \cite{SukochevZanin2022}. Every smooth background used in the solvable sector satisfies this hypothesis.

The same normalized operator enters the fixed-background ghost propagator. Its diagonal resolvent matrix elements give the ghost dressing at a fixed background, and the quadratic Born term is the first nonvanishing contribution to the color-averaged form factor. In four-dimensional Landau gauge, refined Gribov-Zwanziger theory, functional approaches, and lattice calculations support a decoupling pattern with an infrared-finite gluon propagator and a ghost dressing that is not generically enhanced \cite{DudalGraceySorellaVandersickelVerschelde2008,SternbeckIlgenfritzMuellerPreussker2005,EichmannPawlowskiSilva2021}. A pole of $(1+\KA)^{-1}$ at a fixed horizon configuration is a configuration-resolved spectral property, whereas the quantum propagators depend on the full gauge-fixed measure, including the Faddeev-Popov determinant, correlations among configurations, the treatment of Gribov copies, and the infrared and infinite-volume limits.

For a periodic transverse $SU(2)$ background, the spectral problem can be solved explicitly. Embedded-Abelian configurations of this class on a two-torus were studied by Orland and Semenoff in their analysis of extremal curves in Yang-Mills orbit space, including the associated Faddeev-Popov quadratic form and its intersections with the Gribov horizon \cite{OrlandSemenoff2000}. Up to normalization and a global color rotation, the background used here is the $q_2=0$ specialization of their family, now analyzed on the reduced ghost space so that the threshold is the first nontrivial crossing rather than a background-independent zero mode. For this family the eigenvalue problem becomes a Mathieu recurrence, the exact threshold can be compared with its Born approximation and with the corrections of a Feshbach reduction, and the volume dependence of the critical amplitude and classical action is available in closed form. These configurations provide solvable reference backgrounds, while their statistical weight in the Yang-Mills measure requires dynamical information beyond the single-background analysis \cite{SternbeckIlgenfritzMuellerPreussker2005,CucchieriMendes2013,OliveiraPaivaSilva2023,EichmannPawlowskiSilva2021}.

\section{The Faddeev-Popov operator and its Birman-Schwinger congruence}
\label{sec:bs-reduction}

Let $r=N^2-1$ and set
\begin{equation}
  \mathcal H=L^2(T^d;\mathbb C^r),\qquad
  (\Pi_0u)(x)=\frac{1}{|T^d|}\int_{T^d}u(y)\,\dd^dy,
  \qquad P=1-\Pi_0.
  \label{eq:mean-zero-projection}
\end{equation}
The reduced ghost space is $\mathcal H_\perp=P\mathcal H$. The free operator
\begin{equation}
  \Hzero=-\Delta\big|_{\mathcal H_\perp},
  \qquad
  D(\Hzero)=H^2(T^d;\mathbb C^r)\cap\mathcal H_\perp,
  \label{eq:h0-domain}
\end{equation}
is strictly positive. On the torus of side $2\pi$ its lowest eigenvalue is one, and its form domain is
\begin{equation}
  \mathcal Q=H^1(T^d;\mathbb C^r)\cap\mathcal H_\perp.
  \label{eq:form-domain}
\end{equation}
We write $\Lam=\Hzero^{-1/2}$, so $\Lam:\mathcal H_\perp\to\mathcal Q$ is bounded and
\begin{equation}
  \|\nabla\Lam\phi\|_2=\|\phi\|_2.
  \label{eq:lambda-isometry}
\end{equation}

For a smooth transverse background, the adjoint covariant derivative and the Faddeev-Popov operator are
\begin{equation}
  D_i^{ab}=\partial_i\delta^{ab}+g f^{acb}A_i^c,
  \label{eq:adjoint-covariant-derivative}
\end{equation}
\begin{equation}
  \Mfp^{ab}=-\partial^2\delta^{ab}-g f^{acb}A_i^c\partial_i,
  \label{eq:fp-landau}
\end{equation}
where Landau transversality removes the zeroth-order term proportional to $\partial_iA_i$. At lower regularity the second term in \eqref{eq:fp-landau} need not define an $L^2$ operator by itself, so the self-adjoint realization will be defined from its form. For $u,v\in\mathcal Q$ set
\begin{equation}
  h_0[u,v]=\int_{T^d}\partial_i\overline{u^a}\,\partial_i v^a\,\dd^dx,
  \label{eq:h0-form}
\end{equation}
\begin{equation}
  b_A[u,v]
  =-g\int_{T^d}\overline{u^a}\,f^{acb}A_i^c\,\partial_i v^b\,\dd^dx,
  \qquad q_A=h_0+b_A.
  \label{eq:interaction-form}
\end{equation}
Whenever $b_A$ is a bounded form on $\mathcal Q$, the normalized operator is defined by
\begin{equation}
  \langle\phi,\KA\psi\rangle
  =b_A[\Lam\phi,\Lam\psi],
  \qquad \phi,\psi\in\mathcal H_\perp.
  \label{eq:k-form-definition}
\end{equation}
This definition does not require $A_i\partial_i$ to map $H^1$ into $L^2$.

For smooth periodic fields the form definition reproduces the Fourier kernel used later. With $\omega^a(x)=\sum_p e^{ip\cdot x}\omega^a(p)$ and $A_i^a(x)=\sum_qe^{iq\cdot x}A_i^a(q)$,
\begin{equation}
  \Mfp^{ab}(p,k)
  =p^2\delta^{ab}\delta_{p,k}-ig f^{acb}k_iA_i^c(p-k).
  \label{eq:m-kernel}
\end{equation}
Reality and transversality imply $(p_i-k_i)A_i^c(p-k)=0$, hence
\begin{equation}
  \Mfp^{ab}(p,k)
  =p^2\delta^{ab}\delta_{p,k}
   -\frac{ig}{2}f^{acb}(p_i+k_i)A_i^c(p-k),
  \label{eq:m-kernel-symmetric}
\end{equation}
and, on nonzero momenta,
\begin{equation}
  \KA^{ab}(p,k)
  =-ig f^{acb}\frac{k_iA_i^c(p-k)}{|p|\,|k|}.
  \label{eq:k-kernel}
\end{equation}
The symmetric numerator in \eqref{eq:m-kernel-symmetric} makes Hermiticity manifest. Equation \eqref{eq:k-kernel} is used only when the background is regular enough for the Fourier kernel to be interpreted pointwise. The form definition \eqref{eq:k-form-definition} is the primary definition at critical regularity.

\begin{proposition}[Form congruence and Birman-Schwinger criterion]
\label{prop:bs-criterion}
Assume that $q_A$ is a closed, lower-semibounded symmetric form on $\mathcal Q$ and that the form \eqref{eq:k-form-definition} defines a bounded self-adjoint operator $\KA$ on $\mathcal H_\perp$. Let $\Mfp_A$ be the self-adjoint operator associated with $q_A$. Then
\begin{equation}
  q_A[u,v]
  =\big\langle \Hzero^{1/2}u,(1+\KA)\Hzero^{1/2}v\big\rangle,
  \qquad u,v\in\mathcal Q,
  \label{eq:form-congruence}
\end{equation}
and
\begin{equation}
  \Ker\Mfp_A\simeq\Ker(1+\KA),
  \qquad
  n_\pm(\Mfp_A)=n_\pm(1+\KA),
  \qquad
  n_0(\Mfp_A)=n_0(1+\KA).
  \label{eq:inertia}
\end{equation}
If $\KA$ is compact, then
\begin{equation}
  0\in\spec(\Mfp_A)\quad\Longleftrightarrow\quad -1\in\spec(\KA).
  \label{eq:bs-criterion}
\end{equation}
\end{proposition}

\begin{proof}
For $u,v\in\mathcal Q$, set $\phi=\Hzero^{1/2}u$ and $\psi=\Hzero^{1/2}v$. Equation \eqref{eq:lambda-isometry} gives $h_0[u,v]=\langle\phi,\psi\rangle$, while \eqref{eq:k-form-definition} gives $b_A[u,v]=\langle\phi,\KA\psi\rangle$, proving \eqref{eq:form-congruence}. The map $S=\Hzero^{1/2}:\mathcal Q\to\mathcal H_\perp$ is bijective with inverse $\Lam$, so it maps positive, negative, and null subspaces of $q_A$ bijectively onto the corresponding subspaces of the bounded form of $1+\KA$. This proves the inertia identities.

If $u\in\Ker\Mfp_A$, then $q_A[u,v]=0$ for every $v\in\mathcal Q$. Equation \eqref{eq:form-congruence} and the surjectivity of $S$ imply $(1+\KA)Su=0$. Conversely, if $(1+\KA)\eta=0$ and $u=\Lam\eta$, then $q_A[u,v]=0$ for all $v\in\mathcal Q$. The definition of the operator associated with a closed form therefore gives $u\in D(\Mfp_A)$ and $\Mfp_Au=0$. The compact embedding $\mathcal Q\hookrightarrow\mathcal H_\perp$ gives $\Mfp_A$ compact resolvent, so $0\in\spec(\Mfp_A)$ is equivalent to a nontrivial kernel. If $\KA$ is compact, the nonzero spectral value $-1$ can likewise occur only as an eigenvalue. Together with the kernel isomorphism, these facts prove \eqref{eq:bs-criterion}.
\end{proof}

For smooth transverse backgrounds, $\Mfp_A$ coincides with the differential operator \eqref{eq:fp-landau}. The proposition therefore contains the classical elliptic case used in the periodic example without invoking elliptic regularity in the reverse Birman-Schwinger implication. The next section identifies a lower-regularity class in which the form assumptions of the proposition hold.

\section{Quadratic-form realization and self-adjointness at critical regularity}
\label{sec:self-adjointness-scale}

Let $d\ge3$ and let $A\in L^d(T^d;\mathfrak{su}(N)^d)$ satisfy $\partial_iA_i=0$ in the distributional sense. H\"older's inequality and the mean-zero Sobolev embedding give, for $u,v\in\mathcal Q$,
\begin{equation}
  |b_A[u,v]|
  \le C_{d,N}g\,\|A\|_{L^d}
     \|u\|_{L^{2d/(d-2)}}\|\nabla v\|_2
  \le C_{d,N}g\,\|A\|_{L^d}
     \|\nabla u\|_2\|\nabla v\|_2.
  \label{eq:form-bound-holder}
\end{equation}
The estimate makes $b_A$ continuous on the form domain, but closure of $q_A$ does not require the coefficient in \eqref{eq:form-bound-holder} to be smaller than one. For $R>0$, decompose $A=A_{\le R}+A_{>R}$ according to $|A|\le R$ and $|A|>R$. The tail satisfies $\|A_{>R}\|_d\to0$ as $R\to\infty$, while the bounded part obeys
\begin{equation}
  |b_{A_{\le R}}[u,u]|
  \le C_N gR\,\|u\|_2\|\nabla u\|_2
  \le \varepsilon\|\nabla u\|_2^2+C_{\varepsilon,R,A}\|u\|_2^2.
  \label{eq:bounded-part-form}
\end{equation}
Choosing $R$ so that the coefficient of the tail is smaller than $\varepsilon$ yields
\begin{equation}
  |b_A[u,u]|
  \le 2\varepsilon h_0[u]+C_{\varepsilon,A}\|u\|_2^2,
  \qquad u\in\mathcal Q.
  \label{eq:infinitesimal-form-bound}
\end{equation}
Thus $b_A$ is infinitesimally form-bounded with respect to $h_0$. The KLMN theorem and the first representation theorem for closed semibounded forms then make $q_A$ closed and lower semibounded on $\mathcal Q$ and associate with it a unique self-adjoint operator \cite{Kato1995}. Its domain is
\begin{equation}
  D(\Mfp_A)
  =\left\{u\in\mathcal Q:\ \exists f\in\mathcal H_\perp\ \text{such that}\
  q_A[u,v]=\langle f,v\rangle\ \text{for every }v\in\mathcal Q\right\},
  \qquad \Mfp_Au=f.
  \label{eq:fp-form-domain}
\end{equation}
This definition remains meaningful when the first-order term in \eqref{eq:fp-landau} is not an $L^2$ operator. The compact embedding $\mathcal Q\hookrightarrow\mathcal H_\perp$ implies that $\Mfp_A$ has compact resolvent.

Symmetry is obtained without integrating by parts against a rough coefficient. Let $A^{(m)}$ be periodic mollifications of $A$. They remain transverse, converge to $A$ in $L^d$, and satisfy
\begin{equation}
  b_{A^{(m)}}[u,v]=\overline{b_{A^{(m)}}[v,u]}
  \label{eq:smooth-form-symmetry}
\end{equation}
for smooth $u,v$. The bound \eqref{eq:form-bound-holder} extends this identity first to $\mathcal Q$ for each $m$ and then to $A$ by $L^d$ convergence. Hence $q_A$ is symmetric. The same estimate applied to \eqref{eq:k-form-definition}, together with \eqref{eq:lambda-isometry}, gives
\begin{equation}
  \|\KA\|
  \le C_{d,N}g\,\|A\|_{L^d},
  \qquad d\ge3,
  \label{eq:k-norm-bound}
\end{equation}
and the symmetry of $b_A$ makes $\KA$ self-adjoint. A sufficient, though nonsharp, condition for $A$ to lie inside the first Gribov region is therefore $C_{d,N}g\|A\|_{L^d}<1$.

For $d=2$ the Sobolev exponent in \eqref{eq:form-bound-holder} is unavailable. The solvable backgrounds used below are smooth, so the form realization and self-adjointness follow directly from bounded coefficients. At the two-dimensional Orlicz endpoint, the argument concerns compactness and singular-value control of the normalized operator, while the smooth periodic sector admits the differential realization directly from bounded coefficients.

On a smooth transverse core, $\KA$ factors through the Riesz transforms $\Riesz_i=\partial_i\Lam$ and the matrix multiplier $\mathcal A_i=-g\ad(A_i)$,
\begin{equation}
  \KA=\sum_{i=1}^{d}\Lam\,\mathcal A_i\,\Riesz_i.
  \label{eq:k-lambda-riesz}
\end{equation}
For $A\in L^d$, the factorization is first used for smooth transverse approximants and then passed to the form-defined limit \eqref{eq:k-form-definition}.

\section{Cwikel-type estimates for periodic matrix-valued operators}
\label{sec:cwikel}

We use $\mathcal S_{p,\infty}$ for the weak Schatten class with quasi-norm
\begin{equation}
  \|T\|_{\mathcal S_{p,\infty}}:=\sup_{n\ge1} n^{1/p}s_n(T).
  \label{eq:weak-schatten-quasinorm}
\end{equation}

Extend $\Lam$ to the full scalar $L^2(T^d)$ by assigning the Fourier symbol zero at $p=0$ and $|p|^{-1}$ for $p\ne0$. The ordinary torus is the commutative specialization of the torus setting treated by McDonald and Ponge, whose Cwikel estimate gives, for $d>2$,
\begin{equation}
  \|M_f\Lam\|_{\mathcal S_{d,\infty}}
  \le C_d\|f\|_{L^d(T^d)}.
  \label{eq:scalar-cwikel}
\end{equation}
The zero Fourier mode is already absent from the multiplier, so no Euclidean-to-periodic transfer argument is needed \cite{McDonaldPonge2022}. The hypothesis on the Fourier multiplier is consistent with the elementary counting estimate
\begin{equation}
  \#\{p\in\mathbb Z^d\setminus\{0\}:|p|^{-1}>t\}
  \le C_dt^{-d},
  \label{eq:weak-ell-d}
\end{equation}
which places $|p|^{-1}$ in $\ell^{d,\infty}$.

The color extension is finite dimensional and can be separated from the scalar estimate. If $F(x)$ is an $r\times r$ matrix field and $T=M_F\Lam$, then
\begin{equation}
  F(x)^*F(x)\preceq |F(x)|_{\mathrm{op}}^2\,I_r
  \label{eq:matrix-pointwise-order}
\end{equation}
implies
\begin{equation}
  T^*T\preceq
  \big(\Lam M_{|F|_{\mathrm{op}}^2}\Lam\big)\otimes I_r.
  \label{eq:matrix-operator-order}
\end{equation}
Weyl monotonicity then gives
\begin{equation}
  s_n(M_F\Lam)
  \le s_{\lceil n/r\rceil}(M_{|F|_{\mathrm{op}}}\Lam),
  \qquad
  \|M_F\Lam\|_{\mathcal S_{d,\infty}}
  \le r^{1/d}C_d\||F|_{\mathrm{op}}\|_{L^d}.
  \label{eq:matrix-to-scalar}
\end{equation}
Using \eqref{eq:k-lambda-riesz}, invariance of singular values under adjunction, $\|\Riesz_i\|\le1$, and the weak-Schatten ideal property yields the desired estimate.

\begin{theorem}[Periodic matrix-valued Cwikel bound]
\label{thm:cwikel}
Let $d\ge3$ and let $A\in L^d(T^d;\mathfrak{su}(N)^d)$ be transverse in the distributional sense. Then the form-defined operator $\KA$ belongs to $\mathcal S_{d,\infty}$ and
\begin{equation}
  s_n(\KA)
  \le C_{d,N}\,g\,\|A\|_{L^d}\,n^{-1/d},
  \qquad n\ge1.
  \label{eq:cwikel-bound}
\end{equation}
In particular, $\KA$ is compact and the spectral criterion \eqref{eq:bs-criterion} applies to the self-adjoint realization $\Mfp_A$ defined by the form $q_A$.
\end{theorem}

\begin{proof}
For smooth transverse $A$, write $\mathcal A_i=-g\ad(A_i)$ and use \eqref{eq:k-lambda-riesz}. The estimate \eqref{eq:matrix-to-scalar}, the scalar torus estimate \eqref{eq:scalar-cwikel}, $\|\Riesz_i\|\le1$, and $|\ad(A_i)|_{\mathrm{op}}\le\kappa_N|A_i|$ give
\begin{equation}
  \|\KA\|_{\mathcal S_{d,\infty}}
  \le C_{d,N}g\|A\|_{L^d}.
  \label{eq:cwikel-quasinorm}
\end{equation}
For a general transverse $A\in L^d$, take smooth transverse mollifications $A^{(m)}\to A$ in $L^d$. Equation \eqref{eq:cwikel-quasinorm} makes $K_{A^{(m)}}$ Cauchy in $\mathcal S_{d,\infty}$ and therefore in operator norm. Its limit agrees with the form-defined operator \eqref{eq:k-form-definition} by \eqref{eq:form-bound-holder}. Passing to the limit proves the claim.
\end{proof}

\subsection{The two-dimensional endpoint}
\label{sec:cwikel-endpoint}

The scalar endpoint must retain the mean-zero projection explicitly. Let $P_0$ be the projection from scalar $L^2(T^2)$ onto $L^2_0(T^2)$, let $H_0=-\Delta|_{L^2_0}$, and let $\Lam=H_0^{-1/2}$. For $F$ scalar, $V=|F|^2$, and $T=M_F\Lam$, one has
\begin{equation}
  T^*T=\Lam P_0M_VP_0\Lam.
  \label{eq:tstar-t-projected}
\end{equation}
The Birman-Schwinger principle therefore gives the exact counting identity
\begin{equation}
  \#\{n:s_n(M_F\Lam)>\varepsilon\}
  =N_-\!\left(H_0-\varepsilon^{-2}P_0M_VP_0\right).
  \label{eq:bs-clr}
\end{equation}
The compressed and unprojected operators have different negative-eigenvalue counts; they are related instead by the codimension-one estimate
\begin{equation}
  N_-(-\Delta-\varepsilon^{-2}V)-1
  \le
  N_-\!\left(H_0-\varepsilon^{-2}P_0M_VP_0\right)
  \le
  N_-(-\Delta-\varepsilon^{-2}V).
  \label{eq:compressed-count-interlacing}
\end{equation}
Indeed, the intersection of any $m$-dimensional negative subspace of the unprojected quadratic form with $L^2_0$ has dimension at least $m-1$, while every negative subspace of the compressed form is also negative for the full form. Consequently, the compressed problem loses at most one negative direction relative to the unprojected counterexample.

Take
\begin{equation}
  F_L(x)=\sqrt{\frac{\beta}{L}}\frac{1}{|x|}
  \mathbf 1_{\{e^{-L}<|x|<1\}},
  \qquad
  V_L=|F_L|^2,
  \label{eq:counterexample}
\end{equation}
with the annulus contained in a coordinate disk. Its norm remains fixed,
\begin{equation}
  \|F_L\|_2^2=2\pi\beta.
  \label{eq:counterexample-norm}
\end{equation}
Restrict the unprojected quadratic form to radial test functions supported in the annulus and use $t=\log(1/r)\in(0,L)$. The form becomes
\begin{equation}
  2\pi\int_0^L
  \left(|\varphi'(t)|^2-\frac{\beta}{L}|\varphi(t)|^2\right)\dd t.
  \label{eq:radial-form-log}
\end{equation}
The Dirichlet modes $\sin(n\pi t/L)$ are negative whenever $n<\sqrt{\beta L}/\pi$. If
\begin{equation}
  m_L=\#\left\{n\in\mathbb N:\ n<\frac{\sqrt{\beta L}}{\pi}\right\},
  \label{eq:ml-definition}
\end{equation}
then the variational principle gives
\begin{equation}
  N_-(-\Delta-V_L)\ge m_L.
  \label{eq:unprojected-negative-growth}
\end{equation}
Equations \eqref{eq:bs-clr} and \eqref{eq:compressed-count-interlacing} imply
\begin{equation}
  \#\{n:s_n(M_{F_L}\Lam)>1\}
  \ge m_L-1,
  \label{eq:projected-negative-growth}
\end{equation}
and $m_L=\sqrt{\beta L}/\pi+O(1)$. By \eqref{eq:weak-schatten-quasinorm},
\begin{equation}
  \|M_{F_L}\Lam\|_{\mathcal S_{2,\infty}}\ge (m_L-1)^{1/2},
  \label{eq:counterexample-quasinorm-lower}
\end{equation}
while $\|F_L\|_2=(2\pi\beta)^{1/2}$ is independent of $L$. Therefore
\begin{equation}
  \frac{\|M_{F_L}\Lam\|_{\mathcal S_{2,\infty}}}{\|F_L\|_2}
  \longrightarrow\infty.
  \label{eq:counterexample-blowup}
\end{equation}
Thus no uniform $L^2(T^2)\to\mathcal S_{2,\infty}$ bound exists even after the constant mode is removed.

The positive endpoint on the torus is described by an Orlicz condition. Let
\begin{equation}
  M(t)=t\log(e+t),
  \qquad
  \Phi(s)=M(s^2)=s^2\log(e+s^2).
  \label{eq:orlicz-functions}
\end{equation}
Sukochev and Zanin proved the periodic estimate
\begin{equation}
  \left\|(1-\Delta)^{-1/2}M_V(1-\Delta)^{-1/2}\right\|_{\mathcal S_{1,\infty}}
  \le C\|V\|_{L\log L(T^2)}
  \label{eq:sukochev-zanin-endpoint}
\end{equation}
for $V\in L\log L(T^2)$ \cite{SukochevZanin2022}. If $J=(1-\Delta)^{-1/2}$ and $B$ is the bounded Fourier multiplier
\begin{equation}
  b(p)=
  \begin{cases}
    \sqrt{1+|p|^2}/|p|,&p\ne0,\\
    0,&p=0,
  \end{cases}
  \qquad \|B\|\le\sqrt2,
  \label{eq:b-multiplier}
\end{equation}
then the mean-zero inverse square root extended by zero at $p=0$ satisfies $\Lam=BJ$. Since
\begin{equation}
  (M_f\Lam)^*(M_f\Lam)=\Lam M_{|f|^2}\Lam
  =BJM_{|f|^2}JB,
  \label{eq:orlicz-tstar-t}
\end{equation}
the ideal property, \eqref{eq:sukochev-zanin-endpoint}, and $s_n(T^*T)=s_n(T)^2$ give
\begin{equation}
  \|M_f\Lam\|_{\mathcal S_{2,\infty}}
  \le C\|f\|_{L_\Phi(T^2)}.
  \label{eq:orlicz-endpoint}
\end{equation}
Here $L_\Phi$ denotes the Luxemburg Orlicz space associated with $\Phi$. The matrix reduction \eqref{eq:matrix-to-scalar} gives the corresponding finite-color estimate. For smooth transverse backgrounds this combines with the Riesz factorization to yield $\KA\in\mathcal S_{2,\infty}$. The endpoint estimate controls the singular values of the normalized operator and does not require a differential realization of $\Mfp_A$ at the same Orlicz regularity.

\begin{proposition}[Two-dimensional periodic endpoint]
\label{prop:endpoint}
Let $f\in L_\Phi(T^2)$ with $\Phi(s)=s^2\log(e+s^2)$. Then
\begin{equation}
  M_f\Lam\in\mathcal S_{2,\infty},
  \qquad
  \|M_f\Lam\|_{\mathcal S_{2,\infty}}
  \le C\|f\|_{L_\Phi}.
  \label{eq:endpoint-proposition}
\end{equation}
Consequently, every smooth transverse two-dimensional background in the periodic sector considered below produces a compact self-adjoint $\KA\in\mathcal S_{2,\infty}$.
\end{proposition}

No uniform bound of the form $L^{2,1}\to\mathcal S_{2,\infty}$ holds. Let $f_r=r^{-1}\mathbf1_{B_r}$ on a small disk, for which $\|f_r\|_{L^{2,1}}$ is independent of $r$, and define the radial cutoff
\begin{equation}
  u_r(\rho)=
  \begin{cases}
    1,&0\le\rho\le r,\\
    \dfrac{\log(1/\rho)}{\log(1/r)},&r<\rho<1,\\
    0,&\rho\ge1.
  \end{cases}
  \label{eq:logarithmic-cutoff}
\end{equation}
It satisfies $\|\nabla u_r\|_2^2=2\pi/\log(1/r)$ and $\overline{u_r}=O(1/\log(1/r))$. Set $v_r=u_r-\overline{u_r}$ and $\phi_r=H_0^{1/2}v_r$. For sufficiently small $r$, $|v_r|\ge1/2$ on $B_r$, so
\begin{equation}
  \|\phi_r\|_2^2=\frac{2\pi}{\log(1/r)},
  \qquad
  \|M_{f_r}\Lam\phi_r\|_2^2
  =\|f_rv_r\|_2^2\ge\frac{\pi}{4}.
  \label{eq:lorentz-test-vector}
\end{equation}
Consequently,
\begin{equation}
  \|M_{f_r}\Lam\|
  \ge\sqrt{\frac{\log(1/r)}{8}},
  \label{eq:lorentz-counterexample}
\end{equation}
while $\|f_r\|_{L^{2,1}}$ stays bounded. The Orlicz norm in \eqref{eq:orlicz-endpoint} grows at the same logarithmic scale and therefore records the concentration missed by $L^{2,1}$.

\section{The semiclassical crossing law}
\label{sec:crossing-law}

Because $\VA$ is linear in the background and the free factors $\Hzero^{\pm1/2}$ do not depend on $A$, the amplitude rescaling $A\mapsto\lambda A$ propagates directly to $K_{\lambda A}=\lambda\KA$. Writing $\mu_n$ for the negative eigenvalues of $\KA$, an eigenvalue of $K_{\lambda A}$ reaches the horizon value $-1$ when $\lambda\mu_n=-1$, so the number of crossings accumulated up to amplitude $\lambda$ is
\begin{equation}
  N_-(\lambda)=\#\{n:\mu_n\le -\lambda^{-1}\}.
  \label{eq:n-crossing-definition}
\end{equation}
As $\lambda\to\infty$ the threshold $-\lambda^{-1}$ approaches the origin from below, so \eqref{eq:n-crossing-definition} is controlled by the accumulation of the negative eigenvalues of $\KA$ at $0^-$, which is governed by the principal symbol of the compact operator of order $-1$ and the corresponding Weyl law for negative-order pseudodifferential operators \cite{Ponge2023}.

To leading order in the order-$(-1)$ calculus the symbols of $\Hzero^{-1/2}$ and $\VA$ multiply, giving $k^{ab}(x,\xi)=-ig f^{acb}A_i^c(x)\xi_i/|\xi|^2$. For $SU(2)$ the structure constants are $f^{acb}=\epsilon^{acb}$, and introducing the color vector $v^c(x,\hat\xi)=gA_i^c(x)\hat\xi_i$ one writes the color block as $k^{ab}=(i/|\xi|)\Sigma^{ab}$ with $\Sigma^{ab}=\epsilon^{abc}v^c$ real and antisymmetric. Contracting two Levi-Civita symbols gives $(\Sigma^2)^{ab}=v^av^b-|v|^2\delta^{ab}$, and since $(\Sigma v)^a=0$ one has $\Sigma^3=-|v|^2\Sigma$, so the eigenvalues of $\Sigma$ are $\{0,\pm i|v|\}$ and those of the symbol are
\begin{equation}
  \lambda_{\mathrm{sym}}(x,\xi)\in\Big\{0,\ \tfrac{|v(x,\hat\xi)|}{|\xi|},\ -\tfrac{|v(x,\hat\xi)|}{|\xi|}\Big\}.
  \label{eq:symbol-eigenvalues}
\end{equation}
Only the negative branch contributes to \eqref{eq:n-crossing-definition}, and the semiclassical counting function reads
\begin{equation}
  n_-(\varepsilon)\sim\frac{1}{(2\pi)^d}\int\dd^dx\int\dd^d\xi\,\mathbf 1\!\left[\lambda_{\mathrm{sym}}^-(x,\xi)<-\varepsilon\right].
  \label{eq:weyl-phase-space}
\end{equation}
The condition $-|v|/|\xi|<-\varepsilon$ is equivalent to $|\xi|<|v(x,\hat\xi)|/\varepsilon$, so integrating the radius in polar coordinates leaves the angular integral of $|v|^d$,
\begin{equation}
  n_-(\varepsilon)\sim\frac{\varepsilon^{-d}}{d(2\pi)^d}\int\dd^dx\int_{S^{d-1}}\Big(\sum_c[gA_i^c(x)\hat\xi_i]^2\Big)^{d/2}\dd\hat\xi.
  \label{eq:weyl-count-general}
\end{equation}
In two dimensions the angular average is elementary, $\int_{S^1}\hat\xi_i\hat\xi_j\,\dd\hat\xi=\pi\delta_{ij}$, so that $\int_{S^1}|v|^2\dd\hat\xi=\pi|gA|^2$ with $|gA|^2=g^2\sum_{c,i}(A_i^c)^2$, and the constant collapses to $1/(8\pi)$,
\begin{equation}
  n_-(\varepsilon)\sim\frac{1}{8\pi}\,\varepsilon^{-2}\int|gA(x)|^2\,\dd^2x.
  \label{eq:weyl-d2-final}
\end{equation}
Setting $\varepsilon=\lambda^{-1}$ identifies $N_-(\lambda)=n_-(\lambda^{-1})$, whence $N_-(\lambda)\sim(\lambda^2/8\pi)\int|gA|^2$. The singular-value count sees both nonzero branches $\pm|v|/|\xi|$ and therefore carries twice this coefficient, $\#\{n:s_n(\lambda\KA)>1\}\sim(\lambda^2/4\pi)\int|gA|^2$. The factor of two records the passage from the signed negative branch to the unsigned singular values.

For the periodic background \eqref{eq:periodic-background}, the semiclassical
asymptotic can also be tested directly by diagonalizing finite Fourier
truncations of $\KA$ and counting the negative eigenvalues below
$-\varepsilon$. \Cref{fig:semiclassical-counting} compares this counting
function with the prediction \eqref{eq:weyl-d2-final}. In the window
$0.03\le\varepsilon\le0.12$, a log-log fit gives a slope $-2.0077$, close
to the predicted value $-2$. Over the same window, the compensated ratio
$n_-(\varepsilon)\varepsilon^2/\int |gA|^2$ remains close to $1/(8\pi)$,
and the eigenvalue counts are unchanged for Fourier cutoffs
$N=40,60,80$.

\begin{figure}[t]
  \centering
  \includegraphics[width=0.97\linewidth]{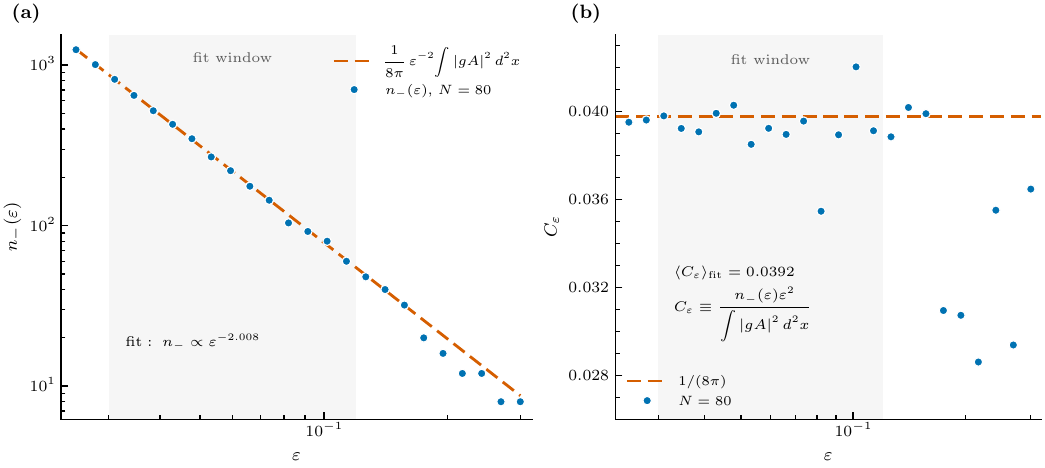}
  \caption{Numerical test of the semiclassical counting law for the periodic
  $SU(2)$ background \eqref{eq:periodic-background} with $g=a=Q=1$.
  (a) Negative-eigenvalue counting function $n_-(\varepsilon)$ obtained from
  the truncated Birman-Schwinger operator at $N=80$, compared with the
  parameter-free asymptotic prediction
  $(8\pi)^{-1}\varepsilon^{-2}\int |gA|^2$.
  (b) Compensated ratio
  $n_-(\varepsilon)\varepsilon^2/\int |gA|^2$, compared with $1/(8\pi)$.
  The shaded region, $0.03\le\varepsilon\le0.12$, is the window used for the
  log-log fit. The eigenvalue counts in this window are identical for
  $N=40$, $60$, and $80$.}
  \label{fig:semiclassical-counting}
\end{figure}

Equation \eqref{eq:weyl-d2-final} is a principal-symbol asymptotic for smooth backgrounds. Extension of this Weyl law to coefficients with only $L^2$ regularity is excluded by the counterexample in \cref{sec:cwikel-endpoint}.

\section{The ghost propagator as a diagonal resolvent}
\label{sec:ghost-resolvent}

At a fixed transverse background $A$ the inverse Faddeev-Popov operator $\Mfp_A^{-1}$ is the ghost Green operator. The quantum ghost propagator follows only after the functional average over the gauge field. Inverting the congruence gives $\Mfp_A^{-1}=\Hzero^{-1/2}(1+\KA)^{-1}\Hzero^{-1/2}$, an identity of bounded operators on $\mathcal H_\perp$ throughout the interior of the Gribov region, where $1+\KA$ is invertible. A plane-wave adjoint state is an eigenstate of $\Hzero^{-1/2}$, so sandwiching the inverse between two such states extracts the free factor cleanly,
\begin{equation}
  \langle k|\Mfp_A^{-1}|k\rangle=\frac{1}{k^2}\,\langle k|(1+\KA)^{-1}|k\rangle,
  \label{eq:ghost-resolvent}
\end{equation}
and the ghost dressing function, defined by stripping the free propagator $1/k^2$, is the diagonal resolvent of the Birman-Schwinger operator at the spectral point $-1$. In the spectral decomposition $\KA=\sum_n\mu_n\ket{\eta_n}\bra{\eta_n}$,
\begin{equation}
  d(k)=\langle k|(1+\KA)^{-1}|k\rangle=\sum_n\frac{|\langle k|\eta_n\rangle|^2}{1+\mu_n},
  \label{eq:dressing-spectral}
\end{equation}
so each mode contributes a positive weight divided by its spectral distance $1+\mu_n$ from the horizon. For a fixed background, a negative eigenvalue near $-1$ can amplify a diagonal ghost matrix element only to the extent that the chosen momentum-color probe overlaps the corresponding eigenvector.

\subsection{Meromorphy in the amplitude and the horizon pole}
\label{sec:amplitude-pole}

The amplitude scaling makes the analytic dependence in \eqref{eq:dressing-spectral} explicit. Since $K_{aA}=a\KA$ is a scalar multiple of a fixed operator, its eigenvectors $\eta_n$ do not depend on the amplitude and only the eigenvalues scale, so the dressing is meromorphic in $a$,
\begin{equation}
  d(k;a)=\sum_n\frac{|\langle k|\eta_n\rangle|^2}{1+a\mu_n},
  \label{eq:dressing-meromorphic}
\end{equation}
with simple poles at $a=-1/\mu_n$ and residues determined by the corresponding eigenvector overlaps. This is the coupling-parameter form of the Birman-Schwinger spectral representation; in the present setting, linearity in the background amplitude keeps the eigenvectors fixed and makes the approach to the horizon directly visible through these poles. Writing $\kappa_{\min}<0$ for the lowest eigenvalue of the unit-amplitude operator, the first horizon crossing occurs at $\acrit=1/|\kappa_{\min}|$, and isolating the critical term gives
\begin{equation}
  d(k;a)=\frac{|\langle k|\eta_{\min}\rangle|^2}{1-a/\acrit}+\sum_{n\neq\min}\frac{|\langle k|\eta_n\rangle|^2}{1+a\mu_n},
  \label{eq:ghost-pole}
\end{equation}
with residue $\operatorname{Res}_{a=\acrit}d(k;a)=-\acrit|\langle k|\eta_{\min}\rangle|^2$, the spectral overlap between the probe momentum and the critical near-zero mode. The resolvent becomes singular as an operator at the horizon, whereas the diagonal element $d(k)$ registers the crossing only when the probe overlaps $\eta_{\min}$. A momentum orthogonal to the critical mode by a residual symmetry of the background keeps $d(k)$ finite while the operator norm of $(1+K_{aA})^{-1}$ already diverges, and when the crossing is degenerate the residue becomes the summed overlap over the critical eigenspace, so degeneracies of the Faddeev-Popov kernel enhance the ghost through the multiplicity of that subspace.

\Cref{fig:ghost_pole} realizes these statements for the periodic $Q=m=1$ reference case, where $\kappa_{\min}\approx-0.5248$ and $\acrit\approx1.9053$. At fixed transverse momentum $m=1$, the diagonal resolvents at chain indices $n=0$ and $n=1$, corresponding to the full momenta $p=(0,1)$ and $p=(1,1)$, rapidly approach the single-pole law $d\simeq R(n)/(1-a/\acrit)$. The $n=2$ probe, $p=(2,1)$, requires a closer approach to the horizon before the pole dominates because its critical overlap is smaller and its regular contribution is comparatively larger. A fit of $d(n;a)=R(n)/(1-a/\acrit)+C_n$ over the eight points closest to the horizon gives residues that agree with $|\langle n|\eta_{\min}\rangle|^2$ to better than $2\times10^{-9}$ relative precision for all three displayed probes.

\begin{figure}[H]
  \centering
  \includegraphics[width=0.85\linewidth]{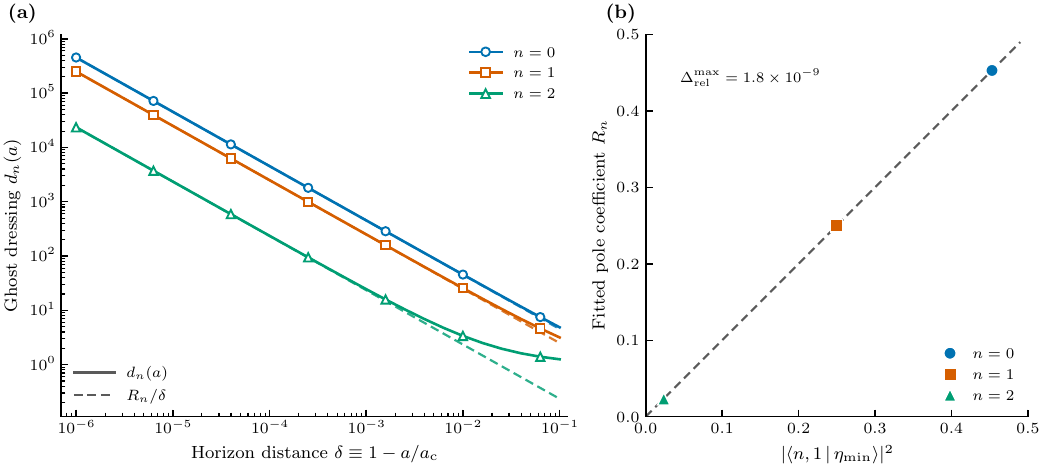}
  \caption{Ghost dressing as the diagonal resolvent of $\KA$ for the periodic $SU(2)$ background at $Q=m=1$. The labels are the chain index $n$ at fixed $m=1$, so $n=0$ denotes the nonzero full momentum $p=(0,1)$ and is not the removed constant ghost mode. Left: $d(n;a)=\langle n,m=1|(1+a\KA)^{-1}|n,m=1\rangle$ as a function of the horizon distance, with dashed curves showing the isolated-pole contribution. Right: fitted pole coefficients $R(n)$ against the critical-state overlaps; the dashed line is equality.}
  \label{fig:ghost_pole}
\end{figure}

\subsection{The quadratic Born term and Gribov's form factor}
\label{sec:gribov-recovery}

Gribov's no-pole construction writes the color-averaged ghost dressing in the form $1/(1-\sigma(k,A))$ and imposes $\sigma(0,A)<1$ as the restriction to $\Omega$ \cite{Gribov1978}. Define
\begin{equation}
  d_{\mathrm{av}}(k,A)=\frac{1}{r}\sum_{a=1}^{r}\langle k,a|(1+\KA)^{-1}|k,a\rangle.
  \label{eq:color-averaged-dressing}
\end{equation}
The term linear in the background vanishes under the color trace because the diagonal kernel of $\KA$ carries $f^{aca}=0$. Since $\KA$ is self-adjoint, the leading nonvanishing contribution is the non-negative quadratic term
\begin{equation}
  d_{\mathrm{av}}(k,A)=1+\frac{1}{r}\sum_{a=1}^{r}\|\KA\ket{k,a}\|^2+O(A^3).
  \label{eq:sigma-recovered}
\end{equation}

\begin{proposition}[Quadratic Born contribution]
\label{prop:gribov-recovery}
For a real transverse $SU(N)$ background with $A_i^c(0)=0$, the color-averaged second-order Born term is
\begin{equation}
\sigma_2(k,A)=\frac{g^2N}{N^2-1}
\sum_{\substack{q\\ k-q\ne0}}
\frac{k_i k_j\,A_i^c(q)A_j^c(-q)}
{k^2\,(k-q)^2}.   
  \label{eq:sigma-explicit}
\end{equation}
the leading nonvanishing contribution to Gribov's ghost form factor in the expansion in powers of the background field.
\end{proposition}

\begin{proof}
Using the kernel \eqref{eq:k-kernel} and reality,
\begin{equation}
  \big[\KA^\dagger\KA\big]^{aa'}(k,k)=g^2\sum_{\ell}\sum_{b,c,c'}f^{bca}f^{bc'a'}\,\frac{k_i k_j\,\overline{A_i^c(\ell)}\,A_j^{c'}(\ell)}{(k+\ell)^2\,k^2}.
  \label{eq:kdagk-raw}
\end{equation}
Setting $a=a'$ and summing collapses the color factor through the adjoint Casimir identity $\sum_{a,b}f^{cab}f^{c'ab}=N\delta^{cc'}$, which contracts $c'=c$, and the substitution $q=-\ell$ with reality brings the result to \eqref{eq:sigma-explicit}. In this step transversality has been used to replace the loop numerator $(k-q)_j$ by $k_j$ through $q_jA_j^c(-q)=0$.
\end{proof}

Under these hypotheses the color-averaged dressing satisfies $d_{\mathrm{av}}(k,A)=1+\sigma_2(k,A)+O(A^3)$. If the exact form factor is defined by $d_{\mathrm{av}}(k,A)=1/(1-\sigma(k,A))$, then $\sigma_2$ is the leading term in the expansion of $\sigma$. The all-order equivalence between Gribov's no-pole condition and Zwanziger's horizon condition concerns their implementation in the functional measure \cite{CapriDudalGuimaraesPalharesSorella2013} and is distinct from any configuration-by-configuration relation between $\sigma_2(k,A)$ and the lowest eigenvalue of $\KA$.

\subsection{The no-pole diagnostic against the exact threshold}
\label{sec:nopole-vs-threshold}

The exact horizon criterion and the quadratic Born diagnostic probe different spectral functionals of the same operator. The first Gribov region is characterized by $\lambda_{\min}(\KA)>-1$, and its boundary by $\lambda_{\min}(\KA)=-1$, whereas the quadratic diagnostic is the probe-dependent second moment
\begin{equation}
  \sigma_2(k,A)
  =\frac{1}{r}\sum_{a=1}^{r}\langle k,a|\KA^2|k,a\rangle
  =\sum_n\mu_n^2\,W_n(k),
  \qquad
  W_n(k)=\frac{1}{r}\sum_{a=1}^{r}|\langle k,a|\eta_n\rangle|^2,
  \label{eq:sigma-second-moment}
\end{equation}
in which the square $\mu_n^2$ discards the sign, so a positive eigenvalue $+\beta$ and a negative eigenvalue $-\beta$ contribute equally although only the negative branch can reach the horizon through $-1$. The color average also makes the normalization explicit: the weight of a single normalized eigenvector at one momentum is $W_n(k)$, not unity. Consequently, $\sigma_2(k,A)$ is not expected to approach one merely because $\mu_{\min}\to-1$. Its relation to the extremal eigenvalue depends on the signed spectrum and on the distribution of the momentum-color overlaps.

The periodic family permits an explicit comparison after projection onto a single charged color sector. The corresponding sector-resolved Born diagonal is not the color-averaged quantity in \eqref{eq:sigma-second-moment}. It couples the probe momentum first to a degenerate pair of nearest neighbors, while the omitted channels enter through a Feshbach self-energy. For $Q=m=1$ this sector-resolved quadratic estimate gives $\anp=2$ against the exact $\acrit\approx1.9053$. The agreement follows from the rapidly convergent Jacobi continued fraction and is not a generic identity between the color-averaged Born moment and the horizon eigenvalue. The low spectrum of $\Mfp_A$ cannot be identified directly with that of $1+\KA$, because the congruence preserves inertia but not numerical eigenvalues. With $\eta=\Hzero^{1/2}\omega$ the Rayleigh quotient of $\Mfp_A$ becomes
\begin{equation}
  \frac{\langle\omega,\Mfp_A\omega\rangle}{\langle\omega,\omega\rangle}=\frac{\langle\eta,(1+\KA)\eta\rangle}{\langle\eta,\Hzero^{-1}\eta\rangle},
  \label{eq:rayleigh-congruence}
\end{equation}
so a small eigenvalue of $\Mfp_A$ can result from a small numerator, from the infrared growth of the denominator through $\|\Hzero^{-1}\|$, or from both. This distinction matters as the volume increases, because the observed drift of the lattice Faddeev-Popov spectrum toward the origin is compatible with eigenvalues of $\KA$ approaching $-1$ but does not establish it without a direct analysis of the normalized operator \cite{SternbeckIlgenfritzMuellerPreussker2005,CucchieriMendes2013}.

The fixed-background pole and the infrared behavior of the ensemble-averaged ghost dressing are distinct: the latter additionally involves the functional average over gauge fields. At fixed $A$ the Green function is $G_A(k)=k^{-2}\langle k|(1+\KA)^{-1}|k\rangle$, while the quantum propagator contains the functional average $\langle G_A(k)\rangle_A$. In four-dimensional Landau gauge, refined Gribov-Zwanziger theory was introduced precisely to accommodate the lattice pattern of a gluon propagator that remains finite at zero momentum together with a ghost propagator that is not infrared enhanced \cite{DudalGraceySorellaVandersickelVerschelde2008}. The gauge-field average depends on the full gauge-fixed measure, including the Faddeev-Popov determinant, correlations among configurations, condensate-induced infrared scales, the distribution of Gribov copies, the overlap residues, and the order of the zero-momentum and infinite-volume limits \cite{VandersickelZwanziger2012,EichmannPawlowskiSilva2021}. The normalized operator instead supplies a configuration-resolved spectral diagnostic through its lowest eigenvalue and the associated momentum-color overlaps.

\section{Periodic \texorpdfstring{$SU(2)$}{SU(2)} backgrounds and the Mathieu reduction}
\label{sec:mathieu}

A tractable example is provided by a periodic transverse $SU(2)$ background in $d=2$, following the use of harmonic Yang-Mills fields in stability analyses \cite{Savvidy1977,BrownWeisberger1979}. We set $g=1$ in this unit-torus spectral calculation and restore the coupling explicitly in the volume-scaling formulas below.
\begin{equation}
  A_x^a=0,\qquad A_y^3=a\cos(Qx),\qquad A_y^1=A_y^2=0.
  \label{eq:periodic-background}
\end{equation}
Up to a global color rotation and normalization this is the $q_2=0$ specialization of the embedded-Abelian periodic family analyzed by Orland and Semenoff in the orbit-space framework, including the associated Faddeev-Popov quadratic form and its relation to the Gribov horizon \cite{OrlandSemenoff2000}. We work on the reduced ghost space with the constant adjoint modes removed, so the threshold is the first nontrivial crossing rather than a background-independent zero mode generated by global gauge rotations. Transversality holds because $A_y^3$ depends on $x$ alone. Only the color-three component is excited, and the perturbation $-f^{a3b}A_y^3\partial_y$ carries $\epsilon^{a3b}$, which vanishes whenever $a$ or $b$ equals three, so the color-three ghost decouples and remains free while the charged components mix. Passing to $\omega^\pm=(\omega^1\pm i\omega^2)/\sqrt2$ diagonalizes the two-by-two block into $(\VA\omega)^\pm=\mp iA_y^3\partial_y\omega^\pm$, the two charged sectors being complex conjugates with a common real spectrum, and retaining the plus sector gives $\Mfp_+=-\partial^2-iA_y^3\partial_y$.

The residual translation invariance in $y$ separates the transverse coordinate: writing $\omega^+(x,y)=e^{imy}\phi(x)$ replaces $\partial_y$ by $im$, turns $-iA_y^3\partial_y$ into $mA_y^3$, and reduces the zero-mode equation to a one-dimensional Schr\"odinger problem with a cosine potential,
\begin{equation}
  \left[-\frac{d^2}{dx^2}+m^2+ma\cos(Qx)\right]\phi(x)=0.
  \label{eq:mathieu-plus}
\end{equation}
The Fourier substitution $\phi(x)=\sum_n\omega_ne^{inx}$ converts this into the three-term recurrence
\begin{equation}
  (n^2+m^2)\omega_n+\frac{ma}{2}(\omega_{n-Q}+\omega_{n+Q})=0,
  \label{eq:mathieu-chain}
\end{equation}
in which the cosine couples each site only to its neighbors at distance $Q$. The recurrence splits into residue classes modulo $Q$, each defining a two-sided Jacobi chain indexed by $j\in\mathbb Z$. In the even sector used below, reflection symmetry identifies the amplitudes at $\pm jQ$ and reduces the problem to a half-line Jacobi operator. The mode $m=0$ yields only the constant solution already removed. For $|m|\ge2$, the operator is bounded below by $m^2-|m|a$. Since the $|m|=1$ crossing occurs at $\acrit\approx1.9053<2$, all sectors with $|m|\ge2$ remain positive at that amplitude. The first nontrivial crossing therefore occurs in the $|m|=1$ sector.

For $Q=1$ and $m=1$ three independent computations agree to numerical resolution: direct diagonalization of \eqref{eq:mathieu-chain}, construction of $\KA$ followed by extraction of its lowest eigenvalue, and bisection of $\Mfp(a)$ on the mode. They return
\begin{equation}
  \acrit=\frac{1}{|\kappa_{\min}|},\qquad \kappa_{\min}\approx-0.5248,\qquad \acrit\approx1.9053,
  \label{eq:acrit-q-one}
\end{equation} Rescaling $x=2u$ carries \eqref{eq:mathieu-plus} at $Q=m=1$ into the canonical Mathieu form $\phi_{uu}+(\alpha-2q\cos 2u)\phi=0$ with $\alpha=-4$ and $q=2a$, and the torus period $2\pi$ in $x$ becomes the period $\pi$ in $u$ carried by the even solution $\mathrm{ce}_0$, so the critical amplitude is fixed by the implicit condition
\begin{equation}
  a_0(2\acrit)=-4,
  \label{eq:mathieu-characterization}
\end{equation}
where $a_0$ is the lowest even Mathieu characteristic value \cite{Mathieu1868,McLachlan1947,MeixnerSchafke1954,AbramowitzStegun1964}. This is an exact special-function characterization of $\acrit$, independent of any chain truncation. A direct evaluation of \eqref{eq:mathieu-characterization} returns $2\acrit\approx3.8106$, in agreement with \eqref{eq:acrit-q-one}.

\subsection{Numerical implementation and convergence}
\label{sec:numerical-reproducibility}

Numerical values for the one-dimensional periodic sector are obtained from the Fourier truncation $n=-N,\ldots,N$ of \eqref{eq:mathieu-chain}. At fixed $Q=m=1$ the charged block is a real symmetric tridiagonal matrix of dimension $2N+1$, with diagonal entries $n^2+1$ and nearest-neighbor entries $a/2$. The neutral color-three sector is free and is not included in this block. The threshold is located by bracketing the zero of its lowest eigenvalue in $1.5<a<2.2$ and applying bisection or Brent iteration until the change in $a$ is below $10^{-13}$. The normalized operator is constructed independently as $K_N=H_{0,N}^{-1/2}V_NH_{0,N}^{-1/2}$ at unit amplitude, so that a second estimate is $a_{\mathrm c,N}=-1/\kappa_{\min,N}$. The Mathieu value is obtained from the root of $a_0(2a)+4=0$ and provides an independent benchmark. The convergence is shown in \cref{tab:numerical-convergence}.

\begin{table}[H]
\centering
\caption{Convergence of the $Q=m=1$ horizon threshold under the truncation $|n|\le N$. Reported digits are restricted to those stable under increasing $N$ and consistent with the Mathieu characteristic value.}
\label{tab:numerical-convergence}
\begin{tabular}{cccc}
\toprule
$N$ & $a_{\mathrm c,N}$ from $\Mfp$ & $\kappa_{\min,N}$ from $K_N$ & $-1/\kappa_{\min,N}$ \\
\midrule
4  & 1.905316059220 & -0.524847305601 & 1.905316059220 \\
6  & 1.905316041156 & -0.524847310577 & 1.905316041156 \\
8  & 1.905316041156 & -0.524847310577 & 1.905316041156 \\
10 & 1.905316041156 & -0.524847310577 & 1.905316041156 \\
20 & 1.905316041156 & -0.524847310577 & 1.905316041156 \\
40 & 1.905316041156 & -0.524847310577 & 1.905316041156 \\
\bottomrule
\end{tabular}
\end{table}

Unless otherwise stated, the numerical calculations employ $N=40$. Repeating the threshold calculation with $N=20,30,40$ changes $a_{\mathrm c}$ by less than $10^{-12}$, which gives a conservative truncation uncertainty below $10^{-12}$ for $a_{\mathrm c}$. The supplementary numerical package reproduces the threshold, the normalized eigenvalue, the Mathieu benchmark, the large-$Q$ checks, the semiclassical counting test, and the diagonal-resolvent data.

The large-$Q$ behavior is governed by a degenerate three-site reduction. Keeping the sites $\{-Q,0,Q\}$ and using the even solution $\omega_{-Q}=\omega_Q=\phi$, $\omega_0=\psi$, the recurrence closes as $(Q^2+m^2)\phi+\tfrac{ma}{2}\psi=0$ and $m^2\psi+ma\phi=0$, the second equation carrying the factor $ma$ rather than $ma/2$ because both neighboring modes contribute to the equation for $n=0$. Eliminating the $n=0$ site gives
\begin{equation}
  a_3^2=2(Q^2+m^2),\qquad a_3\sim\sqrt2\,Q,
  \label{eq:three-level-result}
\end{equation}
and the converged chain reproduces this coefficient, the ratios $\acrit(Q)/Q$ for $Q=1,2,5,10,20,40$ descending toward $\sqrt2$ from above. For $Q>1$ the residue-class chain is truncated as $n=jQ$ with $|j|\le40$. Repeating the calculation with $|j|\le30$ changes the quoted thresholds by less than $3\times10^{-9}$ for the displayed range $Q\le40$. Retaining only the pair $\{0,Q\}$ halves the coupling of the $n=0$ site and produces the incorrect leading coefficient $2Q$, missing the degeneracy of the $\pm Q$ pair.

In the single charged sector used for the Mathieu reduction, the sector-resolved Born diagonal retains exactly the loop connecting the probe $p=(0,m)$ to its two neighbors $k_\pm=(\pm Q,m)$. Since $A_y^3(\pm(Q,0))=a/2$ and $|k_\pm|^2=Q^2+m^2$, the two channels give
\begin{equation}
  \sigma_{\mathrm{Born}}(p;A)=\langle p|\KA^\dagger\KA|p\rangle=2\cdot\frac{(a/2)^2m^2}{(Q^2+m^2)m^2}=\frac{a^2}{2(Q^2+m^2)},
  \label{eq:born-sigma}
\end{equation}
This is a sector-resolved quantity defined within one charged block, distinct from the color-averaged $\sigma_2$ of \eqref{eq:sigma-explicit} and \eqref{eq:sigma-second-moment}. The corresponding Born threshold, defined by $\sigma_{\mathrm{Born}}=1$, is $\anp^2=2(Q^2+m^2)$, which coincides with the three-site truncation \eqref{eq:three-level-result}. Both retain only the nearest-neighbor loop and omit returns through more distant sites. Within this charged block, the quadratic Born term corresponds to the nearest-neighbor truncation of the diagonal resolvent expansion. For $Q=m=1$, it gives $\anp=2$, compared with the exact threshold $\acrit\approx1.9053$, and overestimates the critical amplitude by about five percent.

The exact threshold follows from restoring the tail as a Feshbach self-energy. Ordering the even chain as $\omega_0=\psi$, $\omega_{\pm Q}=\phi_1$, $\omega_{\pm2Q}=\phi_2$, and so on, the ratios $r_j=\phi_j/\phi_{j-1}$ obey the recurrence $r_j=-\tfrac{ma}{2}\big[(jQ)^2+m^2+\tfrac{ma}{2}r_{j+1}\big]^{-1}$ for $j\ge2$, and eliminating the tail dresses the $\phi_1$ site with the effective denominator
\begin{equation}
  D_1^{\mathrm{eff}}=Q^2+m^2-\cfrac{(ma/2)^2}{4Q^2+m^2-\cfrac{(ma/2)^2}{9Q^2+m^2-\cdots}},
  \label{eq:d1-eff}
\end{equation}
which is the Schur complement of the half-line Jacobi operator onto the $\{\psi,\phi_1\}$ block. Eliminating $\psi$ through the $n=0$ equation closes the horizon condition as the self-consistent equation $a^2=2D_1^{\mathrm{eff}}(a)$. Deleting the self-energy returns the Born value. Near threshold the continued fraction subtracts a positive contribution, so $D_1^{\mathrm{eff}}<Q^2+m^2$. This places the exact threshold below the Born estimate and explains the direction of the five-percent discrepancy. Truncating at the $\pm2Q$ channels, that is setting $r_3=0$, retains one self-energy term and reduces the horizon condition to a quadratic in $a^2$ whose solution is the closed form
\begin{equation}
  a_5^2=\frac{4(Q^2+m^2)(4Q^2+m^2)}{8Q^2+3m^2}.
  \label{eq:a5}
\end{equation}
At $Q=m=1$ this gives $a_5=2\sqrt{10/11}\approx1.907$, differing from the exact $\acrit\approx1.9053$ by one part in a thousand. Expanding \eqref{eq:a5} at large $Q$ with $m$ fixed,
\begin{equation}
  \frac{a_5}{Q}=\sqrt2+\frac{7\sqrt2}{16}\frac{m^2}{Q^2}-\frac{69\sqrt2}{512}\frac{m^4}{Q^4}+O(Q^{-6}).
  \label{eq:a5-expansion}
\end{equation}
Equation \eqref{eq:a5-expansion} shows that the $\pm Q$ pair fixes the leading coefficient $\sqrt2$, the $\pm2Q$ pair determines the $Q^{-2}$ correction, and each successive pair contributes two additional inverse powers of $Q$.

\section{Infrared volume laws}
\label{sec:infrared-volume}

On a torus of side $L$ the smallest nonzero momentum is $\Qmin=2\pi/L$, and the longest transverse mode available to the periodic family is $A_2^3(x)=a\cos(\Qmin x_1)$. Measuring every coupled momentum in units of $\Qmin$ maps the dimensionful recurrence onto \eqref{eq:mathieu-chain} verbatim, the coupling and the amplitude entering only through the dimensionless combination $\tilde a=ga/\Qmin$. The first horizon crossing sits at the same universal value $\tilde a=\atil\approx1.9053$ found for the unit-side reference case, so the critical amplitude at side $L$ is
\begin{equation}
  \acrit(L)=\frac{\Qmin}{g}\,\atil=\frac{2\pi\atil}{gL}.
  \label{eq:volume-law}
\end{equation}
Within the periodic class this is exact at every $L$, not a large-volume fit, and it states that the amplitude needed to reach the horizon falls as $1/L$ as the box grows.

The associated field strength is Abelian in color, since only the third component is excited and the transverse field has a single nonzero entry, $F_{12}^3=\partial_1 A_2^3=-a\Qmin\sin(\Qmin x_1)$, the commutator term dropping out. The Euclidean action $S=\tfrac14\int(F_{\mu\nu}^a)^2=\tfrac12\int(F_{12}^3)^2$ then evaluates, using $\int\sin^2(\Qmin x_1)\dd^dx=L^d/2$, to $S=\tfrac14a^2\Qmin^2L^d$, and at the critical amplitude \eqref{eq:volume-law} the volume dependence collapses to a single power,
\begin{equation}
  \Scrit(L)=\frac{a_{\mathrm c}^2\Qmin^2L^d}{4}=\frac{4\pi^4\atil^2}{g^2}\,L^{d-4}.
  \label{eq:action-cost}
\end{equation}
The exponent $d-4$ determines the volume dependence of the critical action. Below four dimensions the cost of the longest critical mode vanishes as $L\to\infty$, at four dimensions it is volume-independent at fixed coupling, and above four dimensions it grows without bound. This threshold is set by the action and is distinct from the Sobolev endpoint at $d=2$ that governs the operator estimates \eqref{eq:k-norm-bound} and \eqref{eq:cwikel-bound}.

These scaling laws characterize the prescribed periodic family. Extending them to a counting of critical directions in the Yang-Mills ensemble would require the gauge-fixed measure, including the Faddeev-Popov determinant and correlations among modes and copies.

\section{Relation to lattice spectra and refined Gribov-Zwanziger theory}
\label{sec:status}

Lattice calculations show that low Faddeev-Popov eigenvalues move toward the origin as the volume increases and that their eigenvectors contribute to the ghost propagator \cite{SternbeckIlgenfritzMuellerPreussker2005,CucchieriMendes2013,OliveiraPaivaSilva2023}. The congruence does not convert that observation directly into accumulation of eigenvalues of $\KA$ at $-1$, because numerical eigenvalues are not preserved. Equation \eqref{eq:rayleigh-congruence} shows that the free infrared scale carried by $\Hzero^{-1}$ also affects the small spectrum of $\Mfp_A$. A direct lattice test should therefore normalize the measured Faddeev-Popov operator. After removing the constant adjoint modes one may define
\begin{equation}
  K_U^{\mathrm{lat}}
  =\big(\Hzero^{\mathrm{lat}}\big)^{-1/2}
   \big(\Mfp_U^{\mathrm{lat}}-\Hzero^{\mathrm{lat}}\big)
   \big(\Hzero^{\mathrm{lat}}\big)^{-1/2},
  \label{eq:lattice-bs-operator}
\end{equation}
using the free lattice Laplacian on the same projected ghost space. Its lowest eigenvalue and the momentum-color overlaps of the associated eigenvectors are the quantities directly comparable with \eqref{eq:bs-criterion} and \eqref{eq:dressing-spectral}.

Refined Gribov-Zwanziger theory incorporates dimension-two condensates and yields an infrared-finite gluon propagator together with a ghost propagator without the enhancement of the original Gribov-Zwanziger scaling solution \cite{DudalGraceySorellaVandersickelVerschelde2008,VandersickelZwanziger2012}. Functional and lattice studies support the same qualitative infrared pattern \cite{SternbeckIlgenfritzMuellerPreussker2005,EichmannPawlowskiSilva2021}. A pole of $(1+\KA)^{-1}$ at a fixed horizon configuration is compatible with this behavior because the ensemble propagator involves a nontrivial gauge-field average. At the configuration level, $\KA$ separates the signed spectral distance to the horizon, measured by $1+\lambda_{\min}(\KA)$, from the probe-dependent overlap weights entering the diagonal ghost resolvent and from the unsigned quadratic moment entering the Born diagnostic.

\section{Conclusion}
\label{sec:conclusion}

After the constant adjoint modes are projected out, the Landau-gauge Faddeev-Popov zero-mode problem admits a Birman-Schwinger formulation on a compact torus through an identity of quadratic forms. For $d\ge3$ and transverse $A\in L^d$, the interaction form is infinitesimally form-bounded with respect to the mean-zero Laplacian, which defines a self-adjoint Faddeev-Popov operator with compact resolvent and a bounded self-adjoint normalized operator $\KA$. The map $u\mapsto\Hzero^{1/2}u$ preserves the positive, negative, and null indices, while the numerical spectra of the two operators remain distinct. The first Gribov horizon is therefore characterized by $-1\in\spec(\KA)$, and the periodic Cwikel estimate yields $s_n(\KA)\le C_{d,N}g\|A\|_{L^d}n^{-1/d}$ for $d\ge3$.

In two dimensions, the Birman-Schwinger count involves the Schrödinger form compressed to the mean-zero space. Since this compression has codimension one in the scalar problem, the negative index differs from that of the unprojected operator by at most one, and the failure of the uniform $L^2$ estimate persists. The positive endpoint is controlled by the periodic Cwikel-Solomyak estimate in the Orlicz class $ |f|^2\in L\log L$, equivalently $f\in L_\Phi$ with $\Phi(t)=t^2\log(e+t^2)$. All smooth two-dimensional backgrounds in the solvable sector satisfy this hypothesis.

For the periodic transverse $SU(2)$ family, the charged sector reduces to a Mathieu recurrence and the first nontrivial $Q=m=1$ crossing satisfies $a_0(2\acrit)=-4$. The numerical threshold is stable under increasing Fourier cutoff and agrees with the independent Mathieu benchmark within the quoted truncation error. At fixed $m=1$, the chain index $n=0$ corresponds to the nonzero full momentum $p=(0,1)$. The color-averaged quadratic Born term is a weighted second spectral moment, whereas the sector-resolved Born approximation in the Mathieu block can be compared directly with the exact threshold and its Feshbach corrections.

The fixed-background resolvent is defined before the functional average over gauge configurations, whereas the ensemble propagators are obtained only after averaging with the gauge-fixed measure. Ensemble observables depend on the full measure, including the Faddeev-Popov determinant, correlations among configurations, condensate effects, Gribov-copy selection, and the order of the infrared and infinite-volume limits. This distinction is compatible with the refined Gribov-Zwanziger decoupling scenario. On the lattice, the normalized operator can be tested directly through its lowest eigenvalues and momentum-color overlaps, providing access to configuration-resolved spectral information that is not contained in the low spectrum of the unnormalized Faddeev-Popov operator alone.

The exact volume law and classical-action scaling characterize the periodic family; their ensemble significance remains to be tested. Extending the normalized spectral diagnostic to representative $SU(3)$ lattice ensembles, with explicit control of volume, copy selection, and eigenvector overlaps, provides a direct test of whether the fixed-background threshold mechanism leaves measurable signatures in ensemble data.

\end{document}